\documentclass[11pt]{extarticle}

\usepackage{amssymb,amsfonts,amsthm}
\usepackage{amsmath}
\usepackage{mathrsfs}
\usepackage{ytableau}
\usepackage[top=1in, bottom=1.25in, left=0.75in, right=0.75in]{geometry}
\usepackage{authblk}
\usepackage{hyperref}

\newtheorem{theorem}{Theorem}[section]
\newtheorem{lemma}[theorem]{Lemma}
\newtheorem{corollary}[theorem]{Corollary}

\theoremstyle{definition}
\newtheorem{definition}[theorem]{Definition}
\newtheorem{proposition}[theorem]{Proposition}
\newtheorem{example}[theorem]{Example}
\theoremstyle{remark}
\newtheorem{remark}[theorem]{Remark}

\numberwithin{equation}{section}

\newcommand{\C}{\mathbb{C}}

\DeclareMathOperator{\Span}{span}

\newcommand{\ip}[2]{\left\langle #1,\,#2\right\rangle}

\DeclareMathOperator{\End}{End}

\DeclareMathOperator{\Spec}{Spec}
\DeclareMathOperator{\Ker}{Ker}

\newcommand{\cK}{\mathcal K}
\newcommand{\cW}{\mathcal W}
\newcommand{\cC}{\mathscr C}

\newcommand{\e}{\mathrm e}
\newcommand{\Ugroup}{\mathrm U}
\newcommand{\iSWAP}{\mathrm{iSWAP}}
\newcommand{\IG}{\mathrm{IG}}
\newcommand{\Haar}{\mathrm{Haar}}
\newcommand{\ot}{\otimes}

\hypersetup{
  colorlinks=true,
  linkcolor=blue,
  citecolor=blue,
  urlcolor=blue
}

\begin{document}

\title{iSWAP maximises the second-moment spectral gap \\in random quantum circuits}

\author{Yanying Liang\thanks{College of Mathematics and Informatics, South China Agricultural University, Guangzhou, 510640, China. Email: yyl2022@scau.edu.cn}
\thanks{Institute of Advanced Intelligence and Computing (IAIC), Agency for Science, Technology and Research (A*STAR), 1 Fusionopolis Way, \#16-16 Connexis, Singapore 138632, Republic of Singapore}
\qquad Haoran Zhu\thanks{Division of Mathematical Sciences, Nanyang Technological University, Singapore 637371, Singapore. Email: zhuh0031@e.ntu.edu.sg}}

\date{}

\maketitle

\begin{abstract}
We prove that the $\iSWAP$ gate maximises the spectral gap of the Hermitian second-moment operator on every connected graph with at least three vertices, among all two-local unitary circuit ensembles.  
We further prove that the polyhedral cone defined by asymmetric four-point inequalities is invariant under the transpose of the $\iSWAP$ semigroup, yielding a componentwise comparison certificate for its positive Perron--Frobenius eigenvector.
These results resolve a conjecture of Kong, Li, and Liu. 

\end{abstract}



\section{Introduction}

Unitary designs provide finite or efficiently samplable substitutes for Haar-random unitaries.  They reproduce the low-order polynomial moments of Haar measure and consequently arise throughout quantum information theory \cite{DankertCleveEmersonLivine,GrossAudenaertEisert}, randomised benchmarking \cite{EmersonAlickiZyczkowski}, decoupling \cite{BrownFawzi}, pseudorandom unitary generation \cite{HarrowLow,BrandaoHarrowHorodecki,Haferkamp,HarrowMehraban}, information scrambling \cite{HunterJones} and the study of chaotic quantum dynamics \cite{HoChoi}.  A basic mechanism for generating such designs is a local random quantum circuit: at each step one selects a small set of qubits and applies a randomly sampled local quantum gate.  The convergence of the resulting walk on the unitary group is governed, at each moment order, by the spectrum of an associated moment operator.  This translates a question about randomness generation into a question about a local operator, or equivalently a frustration-free Hamiltonian.

A substantial part of the literature has concentrated on proving that sufficiently long random circuits form approximate designs, and on obtaining bounds with favourable dependence on the number of qubits, the moment order and the geometry of the circuit.  The foundational result of Harrow and Low \cite{HarrowLow} established approximate unitary $2$-designs from local random circuits.  Brand\~ao, Harrow and Horodecki \cite{BrandaoHarrowHorodecki} proved polynomial-depth convergence for arbitrary moment order, and subsequent work considerably sharpened the quantitative bounds and treated a range of geometries and architectures \cite{Haferkamp,HarrowMehraban,BelkinEtAl,MittalHunterJones}.  Complementary approaches relate these moment operators to statistical-mechanical spin and domain-wall models \cite{BrownViola,HunterJones}, and to spectral problems for frustration-free Hamiltonians on general circuit architectures \cite{MittalHunterJones}.

The gate distribution itself matters.  Haar-random two-qubit gates are convenient analytically, but they need not be optimal when the objective is rapid convergence of a prescribed low-order moment.  In realistic architectures, gates also have unequal physical costs, and the interaction graph may be fixed in advance.  It is therefore natural to ask a comparative question rather than only a convergence question: among the available two-qubit gate ensembles on a fixed graph, which one maximises the moment-operator gap?

Kong, Li, and Liu \cite{KongLiLiu} developed a systematic framework for this comparison.  A central device in their work is the \textbf{ironed gadget}, obtained by surrounding a two-qubit gate with independent single-qubit Haar twirls.  At second-moment order, this operation compresses the gate to a four-dimensional local space and makes the dependence on its KAK coordinates explicit.  Their analysis gives extensive evidence that the $\iSWAP$ gate is universally optimal for convergence to unitary $2$-designs.  They establish the comparison in a substantial semidefinite-comparison region and prove complete-graph optimality for Hermitian two-local ensembles.  For a general fixed connected graph, Kong, Li, and Liu conjectured that
$\iSWAP$ maximises the Hermitian second-moment spectral gap among two-local circuit ensembles \cite[Conjecture~1]{KongLiLiu}.

Here and throughout, $T^{\iSWAP}_{2,G}$ denotes the second-moment operator of the homogeneous graph circuit obtained by placing the ironed $\iSWAP$ ensemble on every edge, with uniform edge selection.



\begin{theorem}
\label{thm:main-intro}
Let $G=(V,E)$ be a connected graph with $|V|=n\geqslant3$.  On each edge $e\in E$, let a probability distribution $\nu_e$ of two-qubit unitaries be prescribed, and let the edge itself be chosen uniformly at every step.  Write $\mathcal E=(\nu_e)_{e\in E}$ for the resulting edge-dependent graph-circuit ensemble.  If its second-moment operator $T^{\mathcal E}_{2,G}$ is Hermitian, then
\begin{equation}\label{eq:main-intro}
 \Delta\bigl(T^{\mathcal E}_{2,G}\bigr)
 \leqslant
 \Delta\bigl(T^{\iSWAP}_{2,G}\bigr).
\end{equation}
\end{theorem}


We now describe the main ideas of the proof. 
In the canonical basis of the single-qubit Haar-invariant second-moment space, the gap operator of every ironed two-qubit gate has the form
\begin{equation}\label{eq:intro-two-projectors}
 h_{\alpha,\beta}=\alpha \mathsf P_+ +\beta \mathsf P_-,
\end{equation}
where $\mathsf P_+$ and $\mathsf P_-$ are two fixed orthogonal rank-one projections.  
Writing the KAK coordinates as $k_x,k_y,k_z$, an elementary inequality on
$\cos(4k_x),\cos(4k_y),\cos(4k_z)$ gives the affine envelope
\begin{equation}\label{eq:intro-envelope}
 0\leqslant \alpha\leqslant\frac{10}{9},
 \qquad
 0\leqslant\beta\leqslant 2-\frac35\alpha.
\end{equation}
The $\iSWAP$ point is $(\alpha,\beta)=(10/9,4/3)$.  If $\beta\leqslant4/3$, a direct Loewner-order comparison is sufficient.  The remaining parameter regime has $\beta>4/3$, where increasing the coefficient of $\mathsf P_-$ is accompanied by decreasing the coefficient of $\mathsf P_+$. 

The inequality in \eqref{eq:intro-envelope} nevertheless shows that it is enough to study the upper boundary
$\beta=2-3\alpha/5$.  A non-unitary, ground-space-adapted similarity transform maps the two common zero modes of the local Hamiltonians to the two consensus states $00$ and $11$.  After the two global consensus states are removed, the boundary Hamiltonian yields a $Z$-matrix block indexed by the non-empty proper subsets of $V$.  At the $\iSWAP$ point, denote this block by $M_0$; it is irreducible and nonsingular.  If $\boldsymbol{\ell}>0$ denotes the left Perron--Frobenius eigenvector associated with its least eigenvalue $\gamma_0$, the comparison reduces to proving the componentwise inequalities
\begin{equation}\label{eq:intro-PF-certificate}
 \boldsymbol{\ell}^{\mathsf T}D_e\geqslant0
 \qquad(e\in E),
\end{equation}
where $D_e$ is the derivative of the boundary block with respect to the parameter on edge $e$.

We introduce a cone of non-negative functions on the Boolean cube, with zero boundary values at the two consensus configurations, defined by the inequalities
\begin{equation}\label{eq:intro-four-point}
 \begin{split}
 4f(A\cup\{i\})+f(A\cup\{j\})
 &\geqslant 2f(A)+2f(A\cup\{i,j\}),\\
 f(A\cup\{i\})+4f(A\cup\{j\})
 &\geqslant 2f(A)+2f(A\cup\{i,j\}).
 \end{split}
\end{equation}
The inequalities are imposed for every pair of distinct vertices $i,j$, not merely for edges of $G$.  This enlargement is essential: an update on an edge sharing one endpoint with $\{i,j\}$ naturally produces a four-point expression on the other two endpoints.  We establish a finite collection of local identities showing that the cone in \eqref{eq:intro-four-point} is invariant under the transpose of the $\iSWAP$ semigroup.  Perron--Frobenius asymptotics then place $\boldsymbol{\ell}$ in this cone, and the inequalities in \eqref{eq:intro-four-point} become exactly the components of \eqref{eq:intro-PF-certificate}.

The certificate compares all edge-dependent boundary parameters with the $\iSWAP$ point.  
We shift the relevant block by a sufficiently large scalar to obtain a non-negative matrix and use the Collatz--Wielandt bound.  This gives the comparison of the third Hamiltonian eigenvalue, counted with multiplicity. 
Finally, a local Haar compression maps an arbitrary ensemble in Theorem~\ref{thm:main-intro} to an edge-dependent average of ironed gadgets.  The min--max principle compares the largest non-trivial eigenvalue before and after compression, while the localisation result of Kong, Li, and Liu identifies the restricted $\iSWAP$ eigenvalue with the spectral gap of the full moment operator when $n\geqslant3$.

The paper is organised as follows.  In Section~\ref{sec:moment-operators}, we recall moment operators, graph-circuit ensembles, the single-qubit Haar compression and the relevant min--max formulation of the spectral gap.  In Section~\ref{sec:local-gadgets}, we derive the two-projector form of a second-moment ironed gadget, establish the KAK parameter envelope and reduce the comparison to its upper boundary.  In Section~\ref{sec:M-matrices}, we perform the ground-state transform, construct the transient blocks and prove irreducibility at the $\iSWAP$ point.  In Section~\ref{sec:four-point-cone}, we introduce the all-pairs four-point cone and prove its invariance under the $\iSWAP$ semigroup.  In Section~\ref{sec:PF-comparison}, we derive the componentwise Perron--Frobenius certificate and complete the comparison for all ironed gadgets.  In Section~\ref{sec:main-proof}, we use local Haar compression to prove Theorem~\ref{thm:main-intro}. Appendices~\ref{app:local-calculations} and \ref{app:four-point-identities} supply the local matrix calculations and expanded verifications of the four-point identities.

\section{Moment operators and local Haar compression}\label{sec:moment-operators}

\subsection{Moment operators}

Let $d\geqslant2$, and let $\nu$ be a Borel probability measure on the unitary group $\Ugroup(d)$.  For a positive integer $t$, the \textbf{$t$-th moment representation} is
\begin{equation}\label{eq:moment-representation}
 \rho_t(U)=U^{\ot t}\ot\overline{U}^{\ot t},
 \qquad U\in\Ugroup(d),
\end{equation}
acting on $(\C^d)^{\ot 2t}$.  Here $\overline U$ denotes entrywise complex conjugation in the fixed computational basis.  The corresponding \textbf{moment operator} is
\begin{equation}\label{eq:moment-operator}
 T_{\nu,t}=\int_{\Ugroup(d)}\rho_t(U)\,\mathrm d\nu(U).
\end{equation}
Since each $\rho_t(U)$ is unitary, $T_{\nu,t}$ is a contraction.

Let $T_{\Haar,t}$ denote the moment operator of Haar measure.  It is the orthogonal projection onto
\begin{equation}\label{eq:Haar-fixed-space}
 \cK_t(d)=\{\mathbf x:\rho_t(U)\mathbf x=\mathbf x\text{ for every }U\in\Ugroup(d)\}.
\end{equation}
The measure $\nu$ is an exact unitary $t$-design precisely when $T_{\nu,t}=T_{\Haar,t}$.  Approximate designs can be defined in several norms; for repeated application of a fixed Hermitian moment operator, the asymptotic convergence rate is controlled by its non-trivial spectrum.

If the total Hilbert-space dimension is at least two, Schur--Weyl duality, or equivalently the standard Haar integration formulae \cite{CollinsSniady}, gives
\begin{equation}\label{eq:K-dimension-two}
 \dim \cK_2(d)=2.
\end{equation}
The two invariant vectors correspond, under vectorisation, to the identity and the swap of the two copies.  We write $\cK$ for this two-dimensional global Haar-fixed space.

\begin{definition}\label{def:nontrivial-eigenvalue}
Let $T$ be a Hermitian moment operator that fixes $\cK$ pointwise.  Then $\cK^\perp$ is $T$-invariant.  Its \textbf{largest non-trivial eigenvalue} is
\begin{equation}\label{eq:lambda-star}
 \lambda_\star(T)
 =\max_{\substack{\mathbf x\in\cK^\perp\\ \|\mathbf x\|=1}}
 \ip{\mathbf x}{T\mathbf x}.
\end{equation}
Its \textbf{spectral gap} is
\begin{equation}\label{eq:spectral-gap}
 \Delta(T)=1-\max\{\lambda_\star(T),|\lambda_{\min}(T)|\}.
\end{equation}
\end{definition}

This definition agrees with the usual notation $1-\max\{\lambda_2,|\lambda_{\min}|\}$, where the two universal unit eigenvectors are not counted as distinct non-trivial modes.  If the ensemble has additional fixed vectors in $\cK^\perp$, then $\lambda_\star(T)=1$ and the gap is zero.

\subsection{Two-local graph circuits}

Let $G=(V,E)$ be a finite simple graph with $V=\{1,\ldots,n\}$.  For each edge $e=\{i,j\}$, let $\nu_e$ be a probability measure on $\Ugroup(4)$, and write $\mathcal E=(\nu_e)_{e\in E}$.  One step of the associated graph circuit chooses an edge uniformly and applies a gate sampled from $\nu_e$ to the two qubits incident with that edge.  Its second-moment operator is
\begin{equation}\label{eq:graph-moment}
 T^{\mathcal E}_{2,G}
 =\frac1{|E|}\sum_{e\in E}T_{\nu_e,2}^{(e)},
\end{equation}
where the superscript indicates that the local moment operator acts on edge $e$ and as the identity on all other tensor factors.  The distributions $\nu_e$ are allowed to be different.

Every term in \eqref{eq:graph-moment} fixes the global Haar space $\cK$, so the same is true of their average.  We assume only that the average $T^{\mathcal E}_{2,G}$ is Hermitian.  

\subsection{The local Haar-invariant space}

We use the Hilbert--Schmidt inner product on $\End((\C^2)^{\ot2})$.  Let $I$ be the identity and let $S$ interchange the two tensor copies.  The range of the single-qubit Haar second-moment projector has the orthonormal basis
\begin{equation}\label{eq:u0-u1}
 \mathbf u_0=\frac12 I,
 \qquad
 \mathbf u_1=\frac1{\sqrt3}\left(S-\frac12I\right).
\end{equation}
For a vertex $v$, write
\begin{equation}\label{eq:Wv}
 \cW_v=\Span\{\mathbf u_0, \mathbf u_1\},
\end{equation}
and let $\mathsf P_v$ be the orthogonal projection onto $\cW_v$.  The product
\begin{equation}\label{eq:global-local-Haar-projection}
 \mathsf P_{\cW}=\prod_{v\in V}\mathsf P_v
\end{equation}
projects onto
\begin{equation}\label{eq:W}
 \cW=\bigotimes_{v\in V}\cW_v
 =\Span\{\mathbf u_0,\mathbf u_1\}^{\ot n}.
\end{equation}
The global Haar-fixed space satisfies
\begin{equation}\label{eq:K-subset-W}
 \cK\subseteq\cW.
\end{equation}
We shall identify $\mathbf u_0,\mathbf u_1$ with the computational symbols $|0\rangle,|1\rangle$ whenever matrices on $\cW$ are written explicitly.

\begin{lemma}\label{lem:Haar-compression}
Let $T$ be a Hermitian operator that fixes $\cK$ pointwise, and let $\mathsf P_{\cW}$ be the orthogonal projection in \eqref{eq:global-local-Haar-projection}.  Then the compression
$
\left.\mathsf P_{\cW}T\mathsf P_{\cW}\right|_{\cW}
$
is Hermitian, fixes $\cK$ pointwise, and
\begin{equation}\label{eq:compression-minmax}
 \lambda_\star(T)
 \geqslant
 \lambda_\star\left(
 \left.\mathsf P_{\cW}T\mathsf P_{\cW}\right|_{\cW}
 \right).
\end{equation}
\end{lemma}

\begin{proof}
Hermiticity follows from
$\mathsf P_{\cW}=\mathsf P_{\cW}^*=\mathsf P_{\cW}^2$.
Since $\cK\subseteq\cW$ and both $\mathsf P_{\cW}$ and $T$ act as the identity on $\cK$, the compression fixes $\cK$.  If $\mathbf x\in\cW$, then $\mathsf P_{\cW}\mathbf x=\mathbf x$, and hence
\begin{equation*}
 \left\langle \mathbf x,
 \left.\mathsf P_{\cW}T\mathsf P_{\cW}\right|_{\cW}
 \mathbf x\right\rangle
 =\ip{\mathbf x}{T\mathbf x}.
\end{equation*}
The admissible unit vectors in $\cW\cap\cK^\perp$ form a subset of those in $\cK^\perp$.  Taking maxima in the Rayleigh quotient proves \eqref{eq:compression-minmax}.
\end{proof}

We shall frequently pass from a compressed moment operator to its gap Hamiltonian.

\begin{lemma}\label{lem:third-eigenvalue}
Let $A$ be Hermitian on a finite-dimensional space, suppose that $A$ fixes a two-dimensional subspace $\cK$, and assume that
$H=I-A$ is positive semidefinite.  Write
\begin{equation*}
 \mu_1(H)\leqslant\mu_2(H)\leqslant\mu_3(H)\leqslant\cdots
\end{equation*}
for the eigenvalues of $H$, counted with multiplicity.  Then
\begin{equation}\label{eq:lambda-star-mu3}
 1-\lambda_\star(A)=\mu_3(H).
\end{equation}
\end{lemma}

\begin{proof}
The two vectors in $\cK$ belong to $\Ker H$, so $\mu_1(H)=\mu_2(H)=0$.  On $\cK^\perp$, maximising the eigenvalue of $A=I-H$ is equivalent to minimising the eigenvalue of $H$.  If $H$ has further zero modes in $\cK^\perp$, both sides of \eqref{eq:lambda-star-mu3} are zero.  Otherwise, $\mu_3(H)$ is the first positive eigenvalue.  In either case the min--max principle gives the stated identity.
\end{proof}

\subsection{Ironed two-qubit gadgets}

Let $e=\{i,j\}$.  The local Haar projection on the endpoints of $e$ is
\begin{equation}\label{eq:Pe}
 \mathsf P_e=\mathsf P_i\mathsf P_j.
\end{equation}
For a two-qubit unitary $U$, its \textbf{second-moment ironed gadget} is the compression
\begin{equation}\label{eq:ironed-gadget-definition}
 \mathcal G_2(U)
 =\left.\mathsf P_e\rho_2(U)\mathsf P_e\right|_{\cW_i\ot\cW_j}.
\end{equation}
By Haar invariance, the surrounding projectors absorb the local-unitary factors in a KAK decomposition.  At second-moment order, $\mathcal G_2(U)$ is Hermitian even though $\rho_2(U)$ need not be; this is one of the useful special features of the ironed model \cite{KongLiLiu}.

The compression of a graph-circuit moment operator is an edge-dependent average of these gadgets.

\begin{lemma}\label{lem:compression-is-gadget-average}
For the graph-circuit ensemble in \eqref{eq:graph-moment}, one has
\begin{equation}\label{eq:compressed-decomposition}
 \left.\mathsf P_{\cW} T^{\mathcal E}_{2,G}\mathsf P_{\cW}\right|_{\cW}
 =\frac1{|E|}\sum_{e\in E}
 \int_{\Ugroup(4)}\mathcal G_2(U)^{(e)}\,\mathrm d\nu_e(U).
\end{equation}
In particular, the right-hand side is Hermitian, whether or not the individual uncompressed edge moment operators are Hermitian.
\end{lemma}

\begin{proof}
The projection $\mathsf P_{\cW}$ factors over vertices.  A gate acting on $e=\{i,j\}$ commutes with every $\mathsf P_v$ for $v\notin e$, and those projectors act as the identity on $\cW_v$.  Therefore
\begin{equation*}
 \left.\mathsf P_{\cW} T_{\nu_e,2}^{(e)}\mathsf P_{\cW}\right|_{\cW}
 =\left(\int \left.\mathsf P_e\rho_2(U)\mathsf P_e\right|_{\cW_i\ot\cW_j}\,\mathrm d\nu_e(U)\right)^{(e)}.
\end{equation*}
Linearity of the integral and the finite edge sum gives \eqref{eq:compressed-decomposition}.  Each integrand is Hermitian at $t=2$, so the average is Hermitian.
\end{proof}

\section{The local two-parameter problem}\label{sec:local-gadgets}

\subsection{KAK coordinates and the four-dimensional matrix}

Every two-qubit unitary is, up to a global phase, locally equivalent to
\begin{equation}\label{eq:KAK}
 \exp\bigl(\mathrm i(k_xX\ot X+k_yY\ot Y+k_zZ\ot Z)\bigr),
\end{equation}
where the triple $(k_x,k_y,k_z)$ may be chosen in a Weyl chamber; see \cite{KhanejaBrockettGlaser,KrausCirac,ZhangValaSastryWhaley}.  Because the local-unitary factors are absorbed by Haar invariance, the second-moment gadget depends only on this triple.

Kong, Li, and Liu \cite[Sec.~4.4, Example~1, Eqs.~(27)--(36)]{KongLiLiu} computed its matrix in the ordered basis
\begin{equation}\label{eq:local-basis}
 |00\rangle,\ |01\rangle,\ |10\rangle,\ |11\rangle,
\end{equation}
where $|0\rangle=\mathbf u_0$ and $|1\rangle=\mathbf u_1$.  We record the result in a form adapted to the present proof.

\begin{proposition}\label{prop:local-matrix}
For every two-qubit unitary, the ironed second-moment operator has the matrix
\begin{equation*}
 T(a,c)=
 \begin{pmatrix}
 1&0&0&0\\
 0&1-c-3b&c&\sqrt3b\\
 0&c&1-c-3b&\sqrt3b\\
 0&\sqrt3b&\sqrt3b&a
 \end{pmatrix},
 \qquad
 b=\frac12(1-a).
\end{equation*}
If
\begin{equation*}
 x=\cos(4k_x),\qquad y=\cos(4k_y),\qquad z=\cos(4k_z),
\end{equation*}
then
\begin{equation*}
 \begin{split}
 a&=\frac19(6+xy+yz+zx),\\
 b&=\frac1{18}(3-xy-yz-zx),\\
 c&=\frac1{12}\bigl(3+xy+yz+zx-2x-2y-2z\bigr).
 \end{split}
\end{equation*}
\end{proposition}

\subsection{Two fixed rank-one projections}


\begin{proposition}\label{prop:rank-one-decomposition}
Define the unit vectors
\begin{equation}\label{eq:phi-plus-minus}
 \boldsymbol{\phi}_+
 =\sqrt{\frac3{10}}\bigl(|01\rangle+|10\rangle\bigr)
 -\sqrt{\frac25}|11\rangle,
 \qquad
 \boldsymbol{\phi}_-=\frac{|01\rangle-|10\rangle}{\sqrt2},
\end{equation}
and the orthogonal rank-one projections
$\mathsf P_+=|\boldsymbol{\phi}_+\rangle\langle\boldsymbol{\phi}_+|$,
$\mathsf P_-=|\boldsymbol{\phi}_-\rangle\langle\boldsymbol{\phi}_-|$.
Then the local gap operator $h(a,c)=I-T(a,c)$ is
\begin{equation}\label{eq:h-alpha-beta}
 h(a,c)=\alpha \mathsf P_+ +\beta \mathsf P_-,
\end{equation}
where
\begin{equation}\label{eq:alpha-beta-ab-c}
 \alpha=5b=\frac52(1-a),
 \qquad
 \beta=3b+2c=\frac32(1-a)+2c.
\end{equation}
Consequently,
\begin{equation}\label{eq:local-spectrum}
 \Spec h(a,c)=\{0,0,\alpha,\beta\}.
\end{equation}
If $\alpha,\beta>0$, then
\begin{equation}\label{eq:local-kernel}
 \Ker h(a,c)
 =\Span\left\{|00\rangle,
 |01\rangle+|10\rangle+\sqrt3|11\rangle\right\}.
\end{equation}
\end{proposition}

\begin{proof}
A direct calculation shows that the vectors in \eqref{eq:phi-plus-minus} are orthonormal.  In the basis \eqref{eq:local-basis}, their projections are
\begin{equation*}
 \mathsf P_+=
 \begin{pmatrix}
 0&0&0&0\\
 0&\frac3{10}&\frac3{10}&-\frac{\sqrt3}{5}\\
 0&\frac3{10}&\frac3{10}&-\frac{\sqrt3}{5}\\
 0&-\frac{\sqrt3}{5}&-\frac{\sqrt3}{5}&\frac25
 \end{pmatrix},
 \qquad
 \mathsf P_-=
 \begin{pmatrix}
 0&0&0&0\\
 0&\frac12&-\frac12&0\\
 0&-\frac12&\frac12&0\\
 0&0&0&0
 \end{pmatrix}.
\end{equation*}
Using $\alpha=5b$ and $\beta=3b+2c$, the diagonal $01,01$ entry of $\alpha \mathsf P_++\beta \mathsf P_-$ is $(3/10)\alpha+(1/2)\beta=3b+c$,
its $01,10$ entry is
$(3/10)\alpha-(1/2)\beta=-c$,
the entries coupling $01$ or $10$ to $11$ are $-\sqrt3b$, and its $11,11$ entry is $2b=1-a$.  These are exactly the entries of $I-T(a,c)$.  This proves \eqref{eq:h-alpha-beta} and \eqref{eq:local-spectrum}.  The orthogonal complement of $\Span\{\boldsymbol{\phi}_+,\boldsymbol{\phi}_-\}$ is the two-dimensional space displayed in \eqref{eq:local-kernel}.
\end{proof}

For later use, set
\begin{equation}\label{eq:r-vector}
 \mathbf r=\frac{|0\rangle+\sqrt3|1\rangle}{2}.
\end{equation}
Then the second vector in \eqref{eq:local-kernel} is proportional to
$\mathbf r\ot\mathbf r-|00\rangle/4$, and in particular
\begin{equation}\label{eq:kernel-product-form}
 \Ker h(a,c)=\Span\{|00\rangle,\mathbf r\ot\mathbf r\}
 \qquad(\alpha,\beta>0).
\end{equation}

\begin{example}[The $\iSWAP$ point]\label{ex:iSWAP-parameters}
For the $\iSWAP$ gate, the parameters $a,c$ in Proposition~\ref{prop:local-matrix} are
\begin{equation*}
 a_0=\frac59,
 \qquad
 c_0=\frac13.
\end{equation*}
Hence
\begin{equation*}
 \alpha_0=\frac{10}{9},
 \qquad
 \beta_0=\frac43.
\end{equation*}
\end{example}

\subsection{An affine envelope for all KAK parameters}

The following estimate contains the entire KAK dependence needed below.

\begin{proposition}\label{prop:parameter-envelope}
For every two-qubit unitary, the parameters in Proposition~\ref{prop:rank-one-decomposition} satisfy
\begin{equation}\label{eq:parameter-envelope}
 0\leqslant\alpha\leqslant\frac{10}{9},
 \qquad
 0\leqslant\beta\leqslant 2-\frac35\alpha.
\end{equation}
The same inequalities hold for arbitrary probability mixtures of ironed two-qubit gadgets.
\end{proposition}

\begin{proof}
With $x,y,z$ as in Proposition~\ref{prop:local-matrix}, substitution of the formulas in that proposition into \eqref{eq:alpha-beta-ab-c} gives
\begin{equation}\label{eq:alpha-beta-xyz}
 \alpha=\frac5{18}(3-xy-yz-zx),
 \qquad
 \beta=1-\frac{x+y+z}{3}.
\end{equation}
Since $x,y,z\in[-1,1]$, the multi-affine function
$xy+yz+zx$ attains its extrema on the vertices of the cube.  At those vertices it takes only the values $-1$ and $3$.  Thus
\begin{equation*}
 -1\leqslant xy+yz+zx\leqslant3,
\end{equation*}
which yields $0\leqslant\alpha\leqslant10/9$.  The lower bound on $\beta$ follows from $x+y+z\leqslant3$.

For the affine upper bound, observe that
\begin{equation}\label{eq:key-cube-inequality}
 xy+yz+zx+2x+2y+2z+3
 =(x+1)(y+1)+(y+1)(z+1)+(z+1)(x+1)
 \geqslant0.
\end{equation}
Using \eqref{eq:alpha-beta-xyz}, inequality \eqref{eq:key-cube-inequality} is equivalent to
\begin{equation*}
 \beta\leqslant2-\frac35\alpha.
\end{equation*}
Finally, the region described by \eqref{eq:parameter-envelope} is convex, and the parameters $\alpha$ and $\beta$ are averaged linearly with the gadget matrix.  Probability mixtures therefore satisfy the same bounds.
\end{proof}


\subsection{Graph Hamiltonians and reduction to the upper boundary}

For each edge $e\in E$, let $(\alpha_e,\beta_e)$ be the parameters of the averaged ironed gadget on that edge, and write $\boldsymbol\alpha=(\alpha_e)_{e\in E}$ and $\boldsymbol\beta=(\beta_e)_{e\in E}$.  Define
\begin{equation}\label{eq:graph-H-alpha-beta}
 H_{\boldsymbol\alpha,\boldsymbol\beta}
 =\frac1{|E|}\sum_{e\in E}
 \bigl(\alpha_e\mathsf P_e^++\beta_e\mathsf P_e^-\bigr)
 \quad\text{on }\cW,
\end{equation}
where $\mathsf P_e^\pm$ denotes $\mathsf P_\pm$ acting on edge $e$ and as the identity on all other tensor factors.  The Hamiltonian in \eqref{eq:graph-H-alpha-beta} is positive semidefinite and has the two global Haar modes in its kernel.

For $\boldsymbol\alpha=(\alpha_e)_{e\in E}\in[0,10/9]^E$, define the upper-boundary Hamiltonian
\begin{equation*}
 H_{\boldsymbol\alpha}^{\mathrm{bd}}
 =\frac1{|E|}\sum_{e\in E}
 \left(
 \alpha_e\mathsf P_e^+
 +\left(2-\frac35\alpha_e\right)\mathsf P_e^-
 \right).
\end{equation*}

\begin{proposition}\label{prop:boundary-reduction}
Suppose that every $(\alpha_e,\beta_e)$ satisfies \eqref{eq:parameter-envelope}.  Then
\begin{equation}\label{eq:Loewner-boundary}
 H_{\boldsymbol\alpha,\boldsymbol\beta}
 \preceq
 H_{\boldsymbol\alpha}^{\mathrm{bd}},
\end{equation}
and consequently
\begin{equation}\label{eq:mu3-boundary}
 \mu_3\bigl(H_{\boldsymbol\alpha,\boldsymbol\beta}\bigr)
 \leqslant
 \mu_3\bigl(H_{\boldsymbol\alpha}^{\mathrm{bd}}\bigr).
\end{equation}
\end{proposition}

\begin{proof}
By Proposition~\ref{prop:parameter-envelope}, each coefficient
$2-\frac35\alpha_e-\beta_e$ is non-negative.  Hence
\begin{equation*}
 H_{\boldsymbol\alpha}^{\mathrm{bd}}
 -H_{\boldsymbol\alpha,\boldsymbol\beta}
 =\frac1{|E|}\sum_{e\in E}
 (2-\frac35\alpha_e-\beta_e)\mathsf P_e^-
 \succeq0.
\end{equation*}
This proves \eqref{eq:Loewner-boundary}; \eqref{eq:mu3-boundary} follows from the min--max principle, or equivalently from Weyl monotonicity \cite{HornJohnson}.
\end{proof}

The problem is therefore reduced to the one-parameter local family
\begin{equation}\label{eq:h-boundary}
 h_\alpha
 =\alpha \mathsf P_+ +\left(2-\frac35\alpha\right)\mathsf P_-,
 \qquad
 0\leqslant\alpha\leqslant\alpha_0=\frac{10}{9}.
\end{equation}

\begin{remark}
This reduction is stronger than merely restricting to the region $\beta>4/3$: it places every admissible gadget below, in Loewner order, a member of a single boundary family whose spectral comparison will be proved simultaneously.
\end{remark}

\section{Boundary Hamiltonians and \texorpdfstring{$Z$}{Z}-matrix blocks}\label{sec:M-matrices}

The upper-boundary Hamiltonians in \eqref{eq:h-boundary} remain Hermitian and positive semidefinite, but their two coefficients move in opposite directions as $\alpha$ varies.  The purpose of this section is to replace them by similar matrices whose off-diagonal signs permit Perron--Frobenius arguments.

\subsection{A ground-space-adapted similarity transform}

Recall the vector $\mathbf r$ in \eqref{eq:r-vector}.  Consider the invertible matrix
\begin{equation}\label{eq:B-matrix}
 B=
 \begin{pmatrix}
 1&-1/\sqrt3\\
 0&2/\sqrt3
 \end{pmatrix},
 \qquad
 B^{-1}=
 \begin{pmatrix}
 1&1/2\\
 0&\sqrt3/2
 \end{pmatrix}.
\end{equation}
It was chosen so that
\begin{equation}\label{eq:B-on-ground-states}
 B|0\rangle=|0\rangle,
 \qquad
 B\mathbf r=|1\rangle.
\end{equation}
Thus the two product zero modes $|0\rangle^{\ot n}$ and $\mathbf r^{\ot n}$ are transformed into the two consensus configurations $|0\rangle^{\ot n}$ and $|1\rangle^{\ot n}$.

\begin{proposition}\label{prop:transformed-local-matrix}
For arbitrary $\alpha,\beta$, one has
\begin{equation}\label{eq:transformed-general}
 (B\ot B)(\alpha \mathsf P_++\beta \mathsf P_-)(B\ot B)^{-1}
 =
 \begin{pmatrix}
 0&-\frac{2\alpha}{5}&-\frac{2\alpha}{5}&0\\
 0&\frac{\alpha+\beta}{2}&\frac{\alpha-\beta}{2}&0\\
 0&\frac{\alpha-\beta}{2}&\frac{\alpha+\beta}{2}&0\\
 0&-\frac{2\alpha}{5}&-\frac{2\alpha}{5}&0
 \end{pmatrix}.
\end{equation}
On the boundary $\beta=2-3\alpha/5$, this becomes
\begin{equation}\label{eq:transformed-boundary}
 \widetilde h_\alpha
 =
 \begin{pmatrix}
 0&-\frac{2\alpha}{5}&-\frac{2\alpha}{5}&0\\
 0&1+\frac{\alpha}{5}&\frac{4\alpha}{5}-1&0\\
 0&\frac{4\alpha}{5}-1&1+\frac{\alpha}{5}&0\\
 0&-\frac{2\alpha}{5}&-\frac{2\alpha}{5}&0
 \end{pmatrix}.
\end{equation}
In particular, for $0\leqslant\alpha\leqslant10/9$, every off-diagonal entry of $\widetilde h_\alpha$ is non-positive.
\end{proposition}

\begin{proof}
The formula follows by substituting the explicit projections displayed in the proof of Proposition~\ref{prop:rank-one-decomposition} and multiplying by $B\ot B$ and its inverse.  The calculation is displayed in Appendix~\ref{app:local-calculations}.  Substitution of $\beta=2-3\alpha/5$ gives \eqref{eq:transformed-boundary}.  Its only potentially positive off-diagonal entry is $4\alpha/5-1$, which is non-positive whenever $\alpha\leqslant5/4$.  The required range ends at $10/9<5/4$.
\end{proof}

Although $\widetilde h_\alpha$ is no longer Hermitian, it is similar to the positive semidefinite matrix $h_\alpha$.  Its two consensus columns, indexed by $00$ and $11$, vanish, and its remaining off-diagonal entries have the sign pattern associated with a killed Markov generator.  After the consensus coordinates are removed, this produces a $Z$-matrix block with non-negative real spectrum.

\subsection{The transient block}

For edge parameters $\boldsymbol\alpha=(\alpha_e)_{e\in E}$, define
\begin{equation}\label{eq:boundary-H-transform}
 \widetilde H_{\boldsymbol\alpha}
 =B^{\ot n}H_{\boldsymbol\alpha}^{\mathrm{bd}}(B^{-1})^{\ot n}.
\end{equation}
We identify the standard basis of $(\C^2)^{\ot n}$ with the subsets of $V$: a set $A\subseteq V$ represents the bitstring whose occupied sites are precisely those in $A$.  

Let $ \Omega=2^V\setminus\{\varnothing,V\}$. The zero columns in \eqref{eq:transformed-boundary} imply that the columns of $\widetilde H_{\boldsymbol\alpha}$ indexed by $\varnothing$ and $V$ vanish.  If the basis is ordered as $\{\varnothing,V\}\sqcup\Omega$, then
\begin{equation}\label{eq:block-form}
 \widetilde H_{\boldsymbol\alpha}
 =
 \begin{pmatrix}
 0&*\\
 0&M_{\boldsymbol\alpha}
 \end{pmatrix}.
\end{equation}
We call $M_{\boldsymbol\alpha}$ the \textbf{transient block}.  It is a $Z$-matrix: all of its off-diagonal entries are non-positive.

\begin{lemma}\label{lem:spectrum-transient-block}
For every $\boldsymbol\alpha\in[0,10/9]^E$,
\begin{equation}\label{eq:spectrum-block-union}
 \Spec\bigl(H_{\boldsymbol\alpha}^{\mathrm{bd}}\bigr)
 =\{0,0\}\cup\Spec(M_{\boldsymbol\alpha}),
\end{equation}
with algebraic multiplicities.  Moreover,
\begin{equation}\label{eq:M-spectrum-positive}
 \Spec(M_{\boldsymbol\alpha})\subseteq[0,\infty),
\end{equation}
and
\begin{equation}\label{eq:mu3-min-spectrum-M}
 \mu_3\bigl(H_{\boldsymbol\alpha}^{\mathrm{bd}}\bigr)
 =\min\Spec(M_{\boldsymbol\alpha}).
\end{equation}
\end{lemma}

\begin{proof}
The block-triangular form \eqref{eq:block-form} gives \eqref{eq:spectrum-block-union}.  The original Hamiltonian $H_{\boldsymbol\alpha}^{\mathrm{bd}}$ is a sum of positive semidefinite matrices and is therefore positive semidefinite.  Similarity preserves its spectrum, so every eigenvalue appearing in \eqref{eq:spectrum-block-union} is real and non-negative.  This proves \eqref{eq:M-spectrum-positive}.

There are two explicitly displayed zero eigenvalues in \eqref{eq:spectrum-block-union}.  Hence the third eigenvalue of the full Hamiltonian is the least eigenvalue of the remaining block.  This remains true when $M_{\boldsymbol\alpha}$ itself is singular: in that case both sides of \eqref{eq:mu3-min-spectrum-M} are zero.
\end{proof}


\subsection{The \texorpdfstring{$\iSWAP$}{iSWAP} block}

At the $\iSWAP$ value $\alpha_0=10/9$, the local transformed matrix is
\begin{equation}\label{eq:iSWAP-transformed-local}
 9\widetilde h_{\alpha_0}
 =
 \begin{pmatrix}
 0&-4&-4&0\\
 0&11&-1&0\\
 0&-1&11&0\\
 0&-4&-4&0
 \end{pmatrix}.
\end{equation}
Write
\begin{equation}\label{eq:M0-gamma0}
 M_0=M_{(\alpha_0)_{e\in E}},
 \qquad
 \gamma_0=\min\Spec(M_0).
\end{equation}
We next establish the two facts required for a strict Perron--Frobenius eigenvector: $M_0$ is irreducible and nonsingular.

For $1\leqslant k\leqslant n-1$, the \textbf{$k$-token graph} of $G$ has the $k$-subsets of $V$ as its vertices, with two subsets adjacent when their symmetric difference is an edge of $G$.  It is the state graph of the simple exclusion process with $k$ particles.

\begin{lemma}\label{lem:token-graph-connected}
If $G$ is connected, then its $k$-token graph is connected for every $1\leqslant k\leqslant n-1$.
\end{lemma}

\begin{proof}
This is the standard irreducibility of the finite simple exclusion process on a connected graph; see, for example, \cite{FabilaMonroyEtAl,LevinPeresWilmer}.  We include an induction.  It is enough to replace $G$ by a spanning tree $T$, because every move available in the $k$-token graph of $T$ is also available in that of $G$.  Induct on $|V(T)|$.  Choose a leaf $v$ of $T$, and let $u$ be its neighbour.  The configurations not containing $v$ induce the $k$-token graph of $T-v$, while those containing $v$ induce a copy of the $(k-1)$-token graph of $T-v$.  In the endpoint cases $k=1$ and $k=|V(T)|-1$, one of these two non-empty classes is a singleton and hence connected; all other non-empty classes are connected by the induction hypothesis.  The two classes are joined: choose a $k$-set containing $v$ but not $u$, which is possible for every $1\leqslant k\leqslant |V(T)|-1$, and move the token from $v$ to $u$.  Hence the entire $k$-token graph is connected.
\end{proof}

\begin{lemma}\label{lem:M0-irreducible}
If $G$ is connected and $n\geqslant3$, then $M_0$ is irreducible.
\end{lemma}

\begin{proof}
Consider the directed graph with vertex set $\Omega$ in which there is an arrow $A\to A'$ whenever the $(A',A)$-entry of $-M_0$ is positive.  Let $A\in\Omega$, and choose an edge $\{i,j\}$ crossing the cut $(A,V\setminus A)$, with $i\in A$ and $j\notin A$.  Reading the column indexed by the local disagreement state in \eqref{eq:iSWAP-transformed-local} gives the three arrows
\begin{equation}\label{eq:three-transitions}
 A\longrightarrow A\setminus\{i\},
 \qquad
 A\longrightarrow A\cup\{j\},
 \qquad
 A\longrightarrow (A\setminus\{i\})\cup\{j\},
\end{equation}
whenever the target remains in $\Omega$.

Starting from a set of size greater than one, repeatedly use the first transition in \eqref{eq:three-transitions}.  Connectedness of $G$ guarantees a cut edge at every non-empty proper set, and the process stops at a singleton without leaving $\Omega$.  Starting from a singleton, repeatedly use the second transition to reach a set of any prescribed cardinality between one and $n-1$.  At a fixed cardinality, the third transition moves a token along an edge to an unoccupied vertex.  Lemma~\ref{lem:token-graph-connected} therefore permits movement between any two subsets of that cardinality.  It follows that every vertex of $\Omega$ can reach every other vertex.  Hence the off-diagonal graph of $M_0$ is strongly connected, which is precisely irreducibility.
\end{proof}

\begin{lemma}\label{lem:M0-nonsingular}
If $G$ is connected, then
\begin{equation}\label{eq:iSWAP-kernel-two}
 \Ker H_{(\alpha_0)_{e\in E}}^{\mathrm{bd}}
 =\Span\{|0\rangle^{\ot n},\mathbf r^{\ot n}\}.
\end{equation}
Consequently, $M_0$ is nonsingular and $\gamma_0>0$.
\end{lemma}

\begin{proof}
At the $\iSWAP$ point, both coefficients in \eqref{eq:h-alpha-beta} are strictly positive.  Since the global Hamiltonian is a sum of positive semidefinite local terms,
\begin{equation*}
 \Ker H_{(\alpha_0)_{e\in E}}^{\mathrm{bd}}
 =\bigcap_{e\in E}\Ker h_{\alpha_0,e}.
\end{equation*}
By \eqref{eq:kernel-product-form}, after applying $B^{\ot n}$, the two endpoint bits on every edge must belong to $\Span\{|00\rangle,|11\rangle\}$.  On a connected graph, these local conditions force all bits to be equal.  Thus
\begin{equation*}
 B^{\ot n}\Ker H_{(\alpha_0)_{e\in E}}^{\mathrm{bd}}
 =\Span\{|0\rangle^{\ot n},|1\rangle^{\ot n}\},
\end{equation*}
which is equivalent to \eqref{eq:iSWAP-kernel-two} by \eqref{eq:B-on-ground-states}.  In the block decomposition \eqref{eq:block-form}, the only zero eigenvalues therefore come from the two consensus columns.  Hence $0\notin\Spec(M_0)$, and $\gamma_0>0$.
\end{proof}

The preceding lemmas show that $M_0$ is an irreducible nonsingular $M$-matrix in the standard sense of \cite{BermanPlemmons,Seneta}.  Equivalently, if $c$ is sufficiently large, then $cI-M_0^{\mathsf T}$ is an irreducible non-negative matrix whose Perron root is $c-\gamma_0$.  The Perron--Frobenius theorem therefore gives the following.

\begin{corollary}\label{cor:positive-left-PF}
There is a vector $\boldsymbol{\ell}\in\mathbb R^{\Omega}$, unique up to positive scaling, such that
\begin{equation*}
 \boldsymbol{\ell}(A)>0\quad(A\in\Omega),
 \qquad
 M_0^{\mathsf T}\boldsymbol{\ell}=\gamma_0\boldsymbol{\ell}.
\end{equation*}
\end{corollary}


\section{An invariant cone of four-point inequalities}\label{sec:four-point-cone}

\subsection{Definition of the cone}

We identify an element $f\in\mathbb R^\Omega$ with a function on the non-empty proper subsets of $V$, and extend it to all of $2^V$ by the boundary convention
\begin{equation}\label{eq:zero-boundary-extension}
 f(\varnothing)=f(V)=0.
\end{equation}
For distinct vertices $i,j$ and a set $A\subseteq V\setminus\{i,j\}$, we use the abbreviations
\begin{equation}\label{eq:A-abbreviations}
 Ai=A\cup\{i\},
 \qquad
 Aj=A\cup\{j\},
 \qquad
 Aij=A\cup\{i,j\}.
\end{equation}

\begin{definition}\label{def:four-point-functionals}
The two oriented four-point functionals associated with $i,j,A$ are
\begin{equation}\label{eq:C-i}
 C_{ij}^{i}(A;f)
 =4f(Ai)+f(Aj)-2f(A)-2f(Aij)
\end{equation}
and
\begin{equation}\label{eq:C-j}
 C_{ij}^{j}(A;f)
 =f(Ai)+4f(Aj)-2f(A)-2f(Aij).
\end{equation}
Define the \textbf{cone} $\cC\subseteq\mathbb R^\Omega$ by
\begin{equation}\label{eq:four-point-cone}
 \cC=
 \left\{
 f:\begin{array}{l}
 f(A)\geqslant0\text{ for every }A\in\Omega,\\
 C_{ij}^{i}(A;f)\geqslant0\text{ and }C_{ij}^{j}(A;f)\geqslant0\\
 \text{for all distinct }i,j\in V
 \text{ and }A\subseteq V\setminus\{i,j\}
 \end{array}
 \right\}.
\end{equation}
\end{definition}

The cone is closed, convex and polyhedral.  Its coefficients may at first appear asymmetric.  Their origin will become transparent in Section~\ref{sec:PF-comparison}: the two quantities in \eqref{eq:C-i}--\eqref{eq:C-j} are exactly five times the non-zero components of the parameter derivative tested against the Perron--Frobenius vector.

\subsection{The local transpose generator}

For an edge $e=\{i,j\}$, define the scaled local transpose operator
\begin{equation*}
 L_{ij}=9\widetilde h_{\alpha_0,ij}^{\mathsf T}.
\end{equation*}
The factor nine removes denominators.  If the bits outside $i,j$ form a set $A\subseteq V\setminus\{i,j\}$, then \eqref{eq:iSWAP-transformed-local} gives
\begin{equation}\label{eq:L-on-A-Aij}
 (L_{ij}f)(A)=0,
 \qquad
 (L_{ij}f)(Aij)=0,
\end{equation}
while
\begin{equation}\label{eq:L-on-Ai}
 (L_{ij}f)(Ai)
 =11f(Ai)-f(Aj)-4f(A)-4f(Aij)
\end{equation}
and
\begin{equation}\label{eq:L-on-Aj}
 (L_{ij}f)(Aj)
 =11f(Aj)-f(Ai)-4f(A)-4f(Aij).
\end{equation}
Boundary values are interpreted according to \eqref{eq:zero-boundary-extension}.  On $\Omega$,
\begin{equation}\label{eq:sum-L-M0}
 9|E|M_0^{\mathsf T}=\sum_{\{i,j\}\in E}L_{ij}.
\end{equation}

The following collection of identities is the local algebra behind cone invariance.  The superscript on $C_{jk}^{k}$, for example, indicates that the coefficient four is attached to the configuration containing $k$.

\begin{lemma}\label{lem:four-point-identities}
Let $i,j,k,l$ be pairwise distinct whenever they occur, and let $f$ be any function with the boundary convention \eqref{eq:zero-boundary-extension}.  Then the following identities hold.

\begin{enumerate}
\item If the update edge is $\{i,j\}$, then
\begin{equation}\label{eq:id-same-edge}
 C_{ij}^{i}(A;-L_{ij}f)
 =C_{ij}^{j}(A;f)-11C_{ij}^{i}(A;f).
\end{equation}

\item If the update edge is $\{i,k\}$ and $k\notin A$, then
\begin{equation}\label{eq:id-ik-k-out}
 C_{ij}^{i}(A;-L_{ik}f)
 =C_{jk}^{k}(A;f)
 +4C_{jk}^{k}(Ai;f)
 -9C_{ij}^{i}(A;f).
\end{equation}

\item If the update edge is $\{i,k\}$ and $A=Bk$, then
\begin{equation}\label{eq:id-ik-k-in}
 C_{ij}^{i}(Bk;-L_{ik}f)
 =4C_{jk}^{k}(B;f)
 +C_{jk}^{k}(Bi;f)
 -3C_{ij}^{i}(Bk;f).
\end{equation}

\item If the update edge is $\{j,k\}$ and $k\notin A$, then
\begin{equation}\label{eq:id-jk-k-out}
 C_{ij}^{i}(A;-L_{jk}f)
 =C_{ik}^{i}(A;f)
 +4C_{ik}^{i}(Aj;f)
 -3C_{ij}^{i}(A;f).
\end{equation}

\item If the update edge is $\{j,k\}$ and $A=Bk$, then
\begin{equation}\label{eq:id-jk-k-in}
 C_{ij}^{i}(Bk;-L_{jk}f)
 =4C_{ik}^{i}(B;f)
 +C_{ik}^{i}(Bj;f)
 -9C_{ij}^{i}(Bk;f).
\end{equation}

\item If $\{k,l\}\cap\{i,j\}=\varnothing$, define
\begin{equation}\label{eq:g-disjoint}
 g(D)=C_{ij}^{i}(D;f),
 \qquad D\subseteq V\setminus\{i,j\}.
\end{equation}
Then
\begin{equation}\label{eq:id-disjoint}
 C_{ij}^{i}(A;-L_{kl}f)=(-L_{kl}g)(A).
\end{equation}
\end{enumerate}
The corresponding identities for $C_{ij}^{j}$ are obtained by interchanging $i$ and $j$.
\end{lemma}

\begin{proof}
All six statements follow by substituting \eqref{eq:L-on-A-Aij}--\eqref{eq:L-on-Aj} into \eqref{eq:C-i}.  We verify the first two here, as they display the two mechanisms that occur; a line-by-line expansion of the remaining cases is given in Appendix~\ref{app:four-point-identities}.

For the update on $\{i,j\}$, only the values at $Ai$ and $Aj$ are changed.  Thus
\begin{align*}
 C_{ij}^{i}(A;-L_{ij}f)
 &=-4(L_{ij}f)(Ai)-(L_{ij}f)(Aj)\\
 &=-43f(Ai)-10f(Aj)+20f(A)+20f(Aij).
\end{align*}
On the other hand,
\begin{align*}
 C_{ij}^{j}(A;f)-11C_{ij}^{i}(A;f)
 &=\bigl(f(Ai)+4f(Aj)-2f(A)-2f(Aij)\bigr)\\
 &\quad-11\bigl(4f(Ai)+f(Aj)-2f(A)-2f(Aij)\bigr)\\
 &=-43f(Ai)-10f(Aj)+20f(A)+20f(Aij),
\end{align*}
which proves \eqref{eq:id-same-edge}.

Now suppose that $k\notin A$ and the update is on $\{i,k\}$.  The configuration $A$ has local state $00$ on $i,k$, whereas $Ai$ has local state $10$.  The same is true after adjoining $j$.  Hence
\begin{align*}
 C_{ij}^{i}(A;-L_{ik}f)
 &=-4(L_{ik}f)(Ai)-(L_{ik}f)(Aij)\\
 &=-44f(Ai)+4f(Ak)+16f(A)+16f(Aik)\\
 &\quad-11f(Aij)+f(Ajk)+4f(Aj)+4f(Aijk).
\end{align*}
Expanding the right-hand side of \eqref{eq:id-ik-k-out} gives precisely the same linear combination.  This proves the second identity.  The other shared-endpoint cases are obtained by the same four- or eight-term calculation.  For a disjoint update, the four-point difference in the $i,j$ coordinates commutes with the operator in the $k,l$ coordinates, giving \eqref{eq:id-disjoint}.
\end{proof}

\subsection{Semigroup invariance}

\begin{theorem}\label{thm:cone-invariance}
For every connected graph $G$ and every $t\geqslant0$,
\begin{equation}\label{eq:cone-invariance}
 \e^{-tM_0^{\mathsf T}}\cC\subseteq\cC.
\end{equation}
\end{theorem}

\begin{proof}
It is enough to work with the rescaled differential equation
\begin{equation}\label{eq:rescaled-ODE}
 \frac{\mathrm d}{\mathrm dt}f_t
 =-\sum_{\{i,j\}\in E}L_{ij}f_t,
\end{equation}
by \eqref{eq:sum-L-M0}.  We use a first-exit argument for the finitely many linear inequalities defining $\cC$.

First, the off-diagonal entries of every $-L_{ij}$, $\{i,j\}\in E$, are non-negative.  Therefore, if $f_t(A)=0$ and every coordinate of $f_t$ is non-negative, then
\begin{equation*}
 \frac{\mathrm d}{\mathrm dt}f_t(A)
 =-\sum_{\{i,j\}\in E}(L_{ij}f_t)(A)\geqslant0.
\end{equation*}
Thus no coordinate can be the first defining inequality to cross from non-negative to negative.

Consider next an active four-point constraint
\begin{equation}\label{eq:active-constraint}
 C_{ij}^{i}(A;f_t)=0,
\end{equation}
while all defining inequalities of $\cC$ remain non-negative.  We inspect the contribution of each edge to its derivative.

If the edge is $\{i,j\}$, identity \eqref{eq:id-same-edge} reduces under \eqref{eq:active-constraint} to
\begin{equation*}
 C_{ij}^{i}(A;-L_{ij}f_t)=C_{ij}^{j}(A;f_t)\geqslant0.
\end{equation*}
If the edge shares exactly one endpoint with $\{i,j\}$, one of
\eqref{eq:id-ik-k-out}--\eqref{eq:id-jk-k-in} applies.  In each case, the term proportional to the active functional vanishes, and all remaining terms are defining four-point functionals of $\cC$ with non-negative coefficients.  The contribution is therefore non-negative.

Finally, suppose that the update edge $\{k,l\}$ is disjoint from $\{i,j\}$.  The function $g$ in \eqref{eq:g-disjoint} is non-negative at every argument because $f_t\in\cC$, and the active constraint says that $g(A)=0$.  Since $-L_{kl}$ has non-negative off-diagonal entries,
\begin{equation*}
 (-L_{kl}g)(A)\geqslant0.
\end{equation*}
Identity \eqref{eq:id-disjoint} again gives a non-negative contribution.

Summing over all edges proves that the derivative of an active functional $C_{ij}^{i}$ cannot point out of the cone.  The same conclusion for $C_{ij}^{j}$ follows by interchanging $i,j$.  Therefore no defining inequality can be the first one to become negative, and the solution of \eqref{eq:rescaled-ODE} remains in $\cC$.  Rescaling time yields \eqref{eq:cone-invariance}.
\end{proof}


\begin{lemma}\label{lem:one-in-cone}
Assume $n\geqslant3$, and define
\begin{equation}\label{eq:f0}
 f_0(A)=1\quad(A\in\Omega),
 \qquad
 f_0(\varnothing)=f_0(V)=0.
\end{equation}
Then $f_0\in\cC$.
\end{lemma}

\begin{proof}
Non-negativity is immediate.  If all four sets $A,Ai,Aj,Aij$ belong to $\Omega$, then
\begin{equation*}
 C_{ij}^{i}(A;f_0)=4+1-2-2=1.
\end{equation*}
If $A=\varnothing$, the assumption $n\geqslant3$ ensures that $Aij\neq V$, and
\begin{equation*}
 C_{ij}^{i}(A;f_0)=4+1-0-2=3.
\end{equation*}
If $A=V\setminus\{i,j\}$, then $Aij=V$, and the same value $3$ is obtained.  These are the only boundary cases.  Interchanging $i,j$ gives the identical calculation for $C_{ij}^{j}$.
\end{proof}

We can now transfer the cone inequalities to the Perron--Frobenius eigenvector.

\begin{theorem}\label{thm:PF-in-cone}
Let $\boldsymbol{\ell}>0$ be the vector in Corollary~\ref{cor:positive-left-PF}, extended by $\boldsymbol{\ell}(\varnothing)=\boldsymbol{\ell}(V)=0$.  Then $\boldsymbol{\ell}\in\cC$.  Equivalently, for every pair of distinct vertices $i,j$ and every $A\subseteq V\setminus\{i,j\}$,
\begin{equation}\label{eq:PF-four-point-i}
 4\boldsymbol{\ell}(Ai)+\boldsymbol{\ell}(Aj)
 \geqslant2\boldsymbol{\ell}(A)+2\boldsymbol{\ell}(Aij)
\end{equation}
and
\begin{equation}\label{eq:PF-four-point-j}
 \boldsymbol{\ell}(Ai)+4\boldsymbol{\ell}(Aj)
 \geqslant2\boldsymbol{\ell}(A)+2\boldsymbol{\ell}(Aij).
\end{equation}
\end{theorem}

\begin{proof}
By Lemma~\ref{lem:one-in-cone} and Theorem~\ref{thm:cone-invariance},
\begin{equation*}
 \e^{-tM_0^{\mathsf T}}f_0\in\cC
 \qquad(t\geqslant0).
\end{equation*}
Because $-M_0^{\mathsf T}$ is an irreducible Metzler matrix, its semigroup is strictly positive for every positive time.  Its spectral bound is the simple eigenvalue $-\gamma_0$, and the corresponding positive right eigenvector is $\boldsymbol{\ell}$.  Since $f_0>0$, its coefficient in the Perron eigendirection is strictly positive.  Perron--Frobenius asymptotics therefore give
\begin{equation*}
 \frac{\e^{-tM_0^{\mathsf T}}f_0}
 {\|\e^{-tM_0^{\mathsf T}}f_0\|}
 \longrightarrow
 \frac{\boldsymbol{\ell}}{\|\boldsymbol{\ell}\|}
 \qquad(t\longrightarrow\infty).
\end{equation*}
The cone $\cC$ is closed and homogeneous, so the limit belongs to $\cC$.  The displayed inequalities are exactly its defining four-point constraints.
\end{proof}

\section{The Perron--Frobenius comparison certificate}\label{sec:PF-comparison}

The four-point inequalities in Theorem~\ref{thm:PF-in-cone} were designed to test the derivative of the boundary family.  We now turn this observation into a spectral comparison valid for arbitrary edge-dependent parameters.

\subsection{The componentwise certificate}

Differentiating \eqref{eq:transformed-boundary} with respect to $\alpha$ gives the constant matrix
\begin{equation}\label{eq:D-matrix}
 D=\frac{\partial\widetilde h_\alpha}{\partial\alpha}
 =
 \begin{pmatrix}
 0&-\frac25&-\frac25&0\\
 0&\frac15&\frac45&0\\
 0&\frac45&\frac15&0\\
 0&-\frac25&-\frac25&0
 \end{pmatrix}.
\end{equation}
For an edge $e=\{i,j\}$, let $D_e$ denote the corresponding global operator, restricted to the transient coordinates $\Omega$.

\begin{proposition}\label{prop:PF-certificate}
For every edge $e\in E$,
\begin{equation}\label{eq:PF-certificate}
 D_e^{\mathsf T}\boldsymbol{\ell}\geqslant0
 \quad\text{componentwise},
\end{equation}
or equivalently $\boldsymbol{\ell}^{\mathsf T}D_e\geqslant0$ componentwise as a row vector.
\end{proposition}

\begin{proof}
Fix $e=\{i,j\}$ and an exterior configuration $A\subseteq V\setminus\{i,j\}$.  The columns of $D_e$ corresponding to local states $00$ and $11$ vanish.  For the column indexed by $Aj$, which has local state $01$, equation \eqref{eq:D-matrix} gives
\begin{align*}
 (\boldsymbol{\ell}^{\mathsf T}D_e)_{Aj}
 &=\frac15\bigl(4\boldsymbol{\ell}(Ai)+\boldsymbol{\ell}(Aj)-2\boldsymbol{\ell}(A)-2\boldsymbol{\ell}(Aij)\bigr)=\frac15C_{ij}^{i}(A;\boldsymbol{\ell})\geqslant0.
\end{align*}
Similarly, for the column indexed by $Ai$,
\begin{align*}
 (\boldsymbol{\ell}^{\mathsf T}D_e)_{Ai}
 &=\frac15\bigl(\boldsymbol{\ell}(Ai)+4\boldsymbol{\ell}(Aj)-2\boldsymbol{\ell}(A)-2\boldsymbol{\ell}(Aij)\bigr)=\frac15C_{ij}^{j}(A;\boldsymbol{\ell})\geqslant0.
\end{align*}
The inequalities are supplied by Theorem~\ref{thm:PF-in-cone}.  This accounts for every column of the local block and proves \eqref{eq:PF-certificate}.
\end{proof}


\subsection{Comparison along the boundary}

Let $\boldsymbol\alpha=(\alpha_e)_{e\in E}\in[0,\alpha_0]^E$.
Since $\widetilde h_\alpha$ is affine in $\alpha$, the transient blocks satisfy
\begin{equation}\label{eq:M-alpha-linear}
 M_{\boldsymbol\alpha}
 =M_0-\frac1{|E|}\sum_{e\in E}(\alpha_0-\alpha_e)D_e.
\end{equation}
Applying Proposition~\ref{prop:PF-certificate} to \eqref{eq:M-alpha-linear} yields
\begin{align}
 M_{\boldsymbol\alpha}^{\mathsf T}\boldsymbol{\ell}
 &=M_0^{\mathsf T}\boldsymbol{\ell}
 -\frac1{|E|}\sum_{e\in E}(\alpha_0-\alpha_e)D_e^{\mathsf T}\boldsymbol{\ell}
 \leqslant\gamma_0\boldsymbol{\ell}.
 \label{eq:subeigenvector}
\end{align}
If $M_{\boldsymbol\alpha}$ were known to be irreducible and nonsingular, one could now appeal directly to a Collatz--Wielandt formula for $M$-matrices.
The next lemma gives the comparison in a form which covers all such degeneracies.

\begin{lemma}\label{lem:shifted-CW}
Let $M$ be a real $Z$-matrix whose spectrum is contained in $[0,\infty)$.  Suppose there are $\mathbf x>0$ and $\gamma\geqslant0$ such that
\begin{equation}\label{eq:M-sub-eigenvector}
 M^{\mathsf T}\mathbf x\leqslant\gamma\mathbf x
\end{equation}
componentwise.  Then
\begin{equation}\label{eq:min-spec-CW}
 \min\Spec(M)\leqslant\gamma.
\end{equation}
\end{lemma}

\begin{proof}
Choose
\begin{equation*}
 c>\max\left\{\max_A M_{AA},\,\max\Spec(M)\right\}.
\end{equation*}
Then $cI-M^{\mathsf T}$ is entrywise non-negative, and \eqref{eq:M-sub-eigenvector} gives
\begin{equation*}
 (cI-M^{\mathsf T})\mathbf x\geqslant(c-\gamma)\mathbf x.
\end{equation*}
The elementary Collatz--Wielandt lower bound for a non-negative matrix, valid without irreducibility, gives
\begin{equation}\label{eq:CW-lower}
 \rho(cI-M^{\mathsf T})
 \geqslant
 \min_{A}\frac{((cI-M^{\mathsf T})\mathbf x)(A)}{\mathbf x(A)}
 \geqslant c-\gamma.
\end{equation}
Every eigenvalue of $M$ is real and belongs to $[0,c)$, so
\begin{equation}\label{eq:rho-shifted-M}
 \rho(cI-M^{\mathsf T})=c-\min\Spec(M).
\end{equation}
Combining \eqref{eq:CW-lower} and \eqref{eq:rho-shifted-M} proves \eqref{eq:min-spec-CW}.
\end{proof}

\begin{theorem}\label{thm:boundary-optimality}
For every connected graph $G$ and every edge-dependent choice
$0\leqslant\alpha_e\leqslant10/9$,
\begin{equation*}
 \mu_3\bigl(H_{\boldsymbol\alpha}^{\mathrm{bd}}\bigr)
 \leqslant
 \mu_3\bigl(H_{(\alpha_0)_{e\in E}}^{\mathrm{bd}}\bigr)
 =\gamma_0.
\end{equation*}
\end{theorem}

\begin{proof}
By Lemma~\ref{lem:spectrum-transient-block}, the spectrum of $M_{\boldsymbol\alpha}$ is contained in $[0,\infty)$, and $M_{\boldsymbol\alpha}$ is a $Z$-matrix.  Equation \eqref{eq:subeigenvector}, together with Lemma~\ref{lem:shifted-CW}, gives
\begin{equation*}
 \min\Spec(M_{\boldsymbol\alpha})\leqslant\gamma_0.
\end{equation*}
A second application of Lemma~\ref{lem:spectrum-transient-block} identifies the left-hand side with $\mu_3(H_{\boldsymbol\alpha}^{\mathrm{bd}})$.  At the $\iSWAP$ point, the same lemma and \eqref{eq:M0-gamma0} identify $\gamma_0$ with the third Hamiltonian eigenvalue.
\end{proof}

\subsection{All ironed gadgets}

We may now combine the boundary comparison with the Loewner reduction from Section~\ref{sec:local-gadgets}.

\begin{theorem}\label{thm:gadget-optimality}
Let $G=(V,E)$ be connected with $|V|\geqslant3$.  On each edge, take an arbitrary probability mixture of second-moment ironed two-qubit gadgets, and let
$\bigl.T^{\IG}_{2,G}\bigr|_{\cW}$ be the resulting operator on $\cW$.  Then
\begin{equation}\label{eq:gadget-lambda-comparison}
 \lambda_\star\bigl(\bigl.T^{\IG}_{2,G}\bigr|_{\cW}\bigr)
 \geqslant
 \lambda_\star\bigl(T^{\iSWAP}_{2,G}\big|_{\cW}\bigr).
\end{equation}
Equivalently,
\begin{equation}\label{eq:gadget-mu3-comparison}
 \mu_3\bigl(I_{\cW}-\bigl.T^{\IG}_{2,G}\bigr|_{\cW}\bigr)
 \leqslant\gamma_0.
\end{equation}
The statement remains valid when the gate mixture depends on the edge.
\end{theorem}

\begin{proof}
By Proposition~\ref{prop:parameter-envelope}, the averaged parameters on every edge satisfy
\begin{equation*}
 0\leqslant\alpha_e\leqslant\alpha_0,
 \qquad
 \beta_e\leqslant2-\frac35\alpha_e.
\end{equation*}
Proposition~\ref{prop:boundary-reduction} and Theorem~\ref{thm:boundary-optimality} give
\begin{equation}\label{eq:gadget-chain-mu3}
 \mu_3\bigl(H_{\boldsymbol\alpha,\boldsymbol\beta}\bigr)
 \leqslant
 \mu_3\bigl(H_{\boldsymbol\alpha}^{\mathrm{bd}}\bigr)
 \leqslant\gamma_0.
\end{equation}
Lemma~\ref{lem:third-eigenvalue} converts \eqref{eq:gadget-chain-mu3} into \eqref{eq:gadget-lambda-comparison}.  No step required the parameters to be constant across edges.
\end{proof}

\begin{remark}\label{rem:what-replaces-PSD}
When $\beta_e\leqslant4/3$, the inequality
\begin{equation*}
 \alpha_e\mathsf P_e^++\beta_e\mathsf P_e^-
 \preceq
 \frac{10}{9}\mathsf P_e^++\frac43\mathsf P_e^-
\end{equation*}
holds locally, and Theorem~\ref{thm:gadget-optimality} follows immediately.  The content of Sections~\ref{sec:M-matrices}--\ref{sec:PF-comparison} is precisely the remaining regime $\beta_e>4/3$, where this local ordering fails.  The Perron--Frobenius vector furnishes a global graph-dependent certificate which substitutes for a local Loewner order.
\end{remark}

\section{Proof of the main theorem}\label{sec:main-proof}

We now return to an arbitrary two-local unitary ensemble.  The gadget theorem controls its local Haar compression.  Two results of Kong, Li, and Liu identify the relevant restricted eigenvalue with the spectral gap of the full $\iSWAP$ moment operator when $n\geqslant3$.

\subsection{The \texorpdfstring{$\iSWAP$}{iSWAP} eigenvalue is local-Haar supported}

For clarity, we state only the specialisations of \cite[Lemma~4.13]{KongLiLiu} and \cite[Corollary~4.14]{KongLiLiu} that are required here.  The former localises the largest non-trivial eigenvalue to $\cW$, while the latter identifies the full spectral gap.

\begin{proposition}\label{prop:KLL-localisation}
Let $G$ be any graph on $n$ vertices, and let $T^{\IG}_{2,G}$ be the full second-moment operator obtained by placing the same ironed gadget on every edge.  Let $T^{\IG}_2$ denote the corresponding two-qubit ironed moment operator.  If
\begin{equation}\label{eq:KLL-threshold}
 n\geqslant\frac{2}{1+\lambda_{\min}(T^{\IG}_2)},
\end{equation}
then the largest non-trivial eigenvalue of $T^{\IG}_{2,G}$ is attained in $\cW$, and its full spectral gap is determined by that eigenvalue rather than by the absolute value of its least eigenvalue.

For the $\iSWAP$ gadget,
\begin{equation*}
 \lambda_{\min}(T^{\iSWAP}_2)=-\frac13.
\end{equation*}
Consequently, \eqref{eq:KLL-threshold} holds whenever $n\geqslant3$, and
\begin{equation*}
 \Delta\bigl(T^{\iSWAP}_{2,G}\bigr)
 =1-\lambda_\star\bigl(T^{\iSWAP}_{2,G}\big|_{\cW}\bigr)
 =\gamma_0.
\end{equation*}
\end{proposition}

\begin{remark}\label{rem:role-n-three}
The four-point cone itself also uses $n\geqslant3$, through Lemma~\ref{lem:one-in-cone}.  More fundamentally, $n=2$ is exceptional: Kong, Li, and Liu show that the largest non-trivial eigenvalue of the full $\iSWAP$ moment operator need not lie in $\cW$.  Thus the threshold in Theorem~\ref{thm:main-intro} is not a cosmetic restriction.
\end{remark}

\subsection{Compression of a general ensemble}

By Lemma~\ref{lem:compression-is-gadget-average},
$
\left.\mathsf P_{\cW}T^{\mathcal E}_{2,G}\mathsf P_{\cW}\right|_{\cW}
$
is an edge-dependent average of ironed gadgets.  Therefore Theorem~\ref{thm:gadget-optimality} gives
\begin{equation*}
 \lambda_\star\left(
 \left.\mathsf P_{\cW}T^{\mathcal E}_{2,G}\mathsf P_{\cW}\right|_{\cW}
 \right)
 \geqslant
 \lambda_\star\bigl(T^{\iSWAP}_{2,G}\big|_{\cW}\bigr)
 =1-\gamma_0.
\end{equation*}
On the other hand, Lemma~\ref{lem:Haar-compression} gives
\begin{equation*}
 \lambda_\star\bigl(T^{\mathcal E}_{2,G}\bigr)
 \geqslant
 \lambda_\star\left(
 \left.\mathsf P_{\cW}T^{\mathcal E}_{2,G}\mathsf P_{\cW}\right|_{\cW}
 \right).
\end{equation*}
Combining the two estimates yields
\begin{equation}\label{eq:general-lambda-bound}
 \lambda_\star\bigl(T^{\mathcal E}_{2,G}\bigr)
 \geqslant1-\gamma_0.
\end{equation}

We are now ready to prove the main result.

\begin{proof}[Proof of Theorem~\ref{thm:main-intro}]
The definition of the spectral gap gives
\begin{equation*}
 \Delta\bigl(T^{\mathcal E}_{2,G}\bigr)
 =1-\max\left\{
 \lambda_\star\bigl(T^{\mathcal E}_{2,G}\bigr),
 \left|\lambda_{\min}\bigl(T^{\mathcal E}_{2,G}\bigr)\right|
 \right\}
 \leqslant1-\lambda_\star\bigl(T^{\mathcal E}_{2,G}\bigr).
\end{equation*}
Using \eqref{eq:general-lambda-bound}, followed by the $\iSWAP$ identification in Proposition~\ref{prop:KLL-localisation}, we obtain
\begin{align*}
 \Delta\bigl(T^{\mathcal E}_{2,G}\bigr)
 \leqslant1-\lambda_\star\bigl(T^{\mathcal E}_{2,G}\bigr)
 \leqslant\gamma_0
 =\Delta\bigl(T^{\iSWAP}_{2,G}\bigr).
\end{align*}
This is \eqref{eq:main-intro} and proves the conjecture.
\end{proof}

\begin{corollary}\label{cor:homogeneous-gates}
On every connected graph with at least three vertices, the homogeneous $\iSWAP$ circuit has spectral gap at least that of every homogeneous Hermitian two-qubit gate ensemble, every probability mixture of such ensembles and every edge-dependent Hermitian two-local circuit ensemble.
\end{corollary}

\begin{proof}
Each listed model is a special case of Theorem~\ref{thm:main-intro}.
\end{proof}

\begin{corollary}\label{cor:negative-eigenvalues-no-help}
Allowing a competing Hermitian moment operator to have negative eigenvalues cannot improve upon the $\iSWAP$ gap.
\end{corollary}

\begin{proof}
The comparison \eqref{eq:general-lambda-bound} already forces its largest non-trivial eigenvalue to be at least the $\iSWAP$ value.  The absolute value of a negative eigenvalue can only leave the maximum in \eqref{eq:spectral-gap} unchanged or increase it, and therefore cannot improve the gap.
\end{proof}

\appendix

\section{Local matrix calculations}\label{app:local-calculations}

This appendix records the calculations used in Propositions~\ref{prop:rank-one-decomposition} and \ref{prop:transformed-local-matrix}.  They are included both for reference and to fix all normalisations.

\subsection{The two-projector form}

With the basis ordered as in \eqref{eq:local-basis}, the vectors in \eqref{eq:phi-plus-minus} are
\begin{equation*}
 \boldsymbol{\phi}_+=
 \begin{pmatrix}
 0\\ \sqrt{3/10}\\ \sqrt{3/10}\\-\sqrt{2/5}
 \end{pmatrix},
 \qquad
 \boldsymbol{\phi}_-=
 \begin{pmatrix}
 0\\1/\sqrt2\\-1/\sqrt2\\0
 \end{pmatrix}.
\end{equation*}
Their outer products are the matrices displayed in the proof of Proposition~\ref{prop:rank-one-decomposition}.  Hence
\begin{equation}\label{eq:alphaP-betaP-explicit}
 \alpha \mathsf P_++\beta \mathsf P_-
 =
 \begin{pmatrix}
 0&0&0&0\\
 0&\frac{3\alpha}{10}+\frac\beta2&\frac{3\alpha}{10}-\frac\beta2&-\frac{\sqrt3\alpha}{5}\\
 0&\frac{3\alpha}{10}-\frac\beta2&\frac{3\alpha}{10}+\frac\beta2&-\frac{\sqrt3\alpha}{5}\\
 0&-\frac{\sqrt3\alpha}{5}&-\frac{\sqrt3\alpha}{5}&\frac{2\alpha}{5}
 \end{pmatrix}.
\end{equation}
Substituting $\alpha=5b$ and $\beta=3b+2c$ gives
\begin{equation*}
 \alpha \mathsf P_++\beta \mathsf P_-
 =
 \begin{pmatrix}
 0&0&0&0\\
 0&3b+c&-c&-\sqrt3b\\
 0&-c&3b+c&-\sqrt3b\\
 0&-\sqrt3b&-\sqrt3b&2b
 \end{pmatrix}.
\end{equation*}
Since $2b=1-a$, this is exactly $I-T(a,c)$.

The two displayed kernel vectors may also be checked directly.  Both are orthogonal to $\boldsymbol{\phi}_-$.  Moreover,
\begin{equation*}
 \sqrt{\frac3{10}}(1+1)-\sqrt{\frac25}\sqrt3=0,
\end{equation*}
so $|01\rangle+|10\rangle+\sqrt3|11\rangle$ is orthogonal to $\boldsymbol{\phi}_+$.  Since $\mathsf P_+$ and $\mathsf P_-$ have orthogonal one-dimensional ranges, these vectors span the full common kernel.

\subsection{The similarity transform}

The Kronecker products of the matrices in \eqref{eq:B-matrix} are
\begin{equation}\label{eq:B-tensor-B}
 B\ot B=
 \begin{pmatrix}
 1&-1/\sqrt3&-1/\sqrt3&1/3\\
 0&2/\sqrt3&0&-2/3\\
 0&0&2/\sqrt3&-2/3\\
 0&0&0&4/3
 \end{pmatrix}
\end{equation}
and
\begin{equation*}
 B^{-1}\ot B^{-1}=
 \begin{pmatrix}
 1&1/2&1/2&1/4\\
 0&\sqrt3/2&0&\sqrt3/4\\
 0&0&\sqrt3/2&\sqrt3/4\\
 0&0&0&3/4
 \end{pmatrix}.
\end{equation*}
Multiplying \eqref{eq:alphaP-betaP-explicit} on the left by \eqref{eq:B-tensor-B} and on the right by the displayed matrix $B^{-1}\ot B^{-1}$ gives
\begin{equation*}
 \begin{pmatrix}
 0&-2\alpha/5&-2\alpha/5&0\\
 0&(\alpha+\beta)/2&(\alpha-\beta)/2&0\\
 0&(\alpha-\beta)/2&(\alpha+\beta)/2&0\\
 0&-2\alpha/5&-2\alpha/5&0
 \end{pmatrix},
\end{equation*}
which is \eqref{eq:transformed-general}.  On substituting $\beta=2-3\alpha/5$,
\begin{equation*}
 \frac{\alpha+\beta}{2}=1+\frac\alpha5,
 \qquad
 \frac{\alpha-\beta}{2}=\frac{4\alpha}{5}-1,
\end{equation*}
and \eqref{eq:transformed-boundary} follows.

\subsection{The \texorpdfstring{$\iSWAP$}{iSWAP} derivative}

At $\alpha_0=10/9$,
\begin{equation*}
 -\frac{2\alpha_0}{5}=-\frac49,
 \qquad
 1+\frac{\alpha_0}{5}=\frac{11}{9},
 \qquad
 \frac{4\alpha_0}{5}-1=-\frac19,
\end{equation*}
which gives \eqref{eq:iSWAP-transformed-local}.  Since every entry in \eqref{eq:transformed-boundary} is affine in $\alpha$, differentiation gives \eqref{eq:D-matrix}.

\section{Expanded verification of the four-point identities}\label{app:four-point-identities}

We give additional details for Lemma~\ref{lem:four-point-identities}.  The calculations are local, so vertices not displayed below are held fixed.  For compactness, write
\begin{equation*}
 f_X=f(A\cup X)
\end{equation*}
when $A$ is disjoint from the vertices in $X$.  All identities are polynomial identities in these finitely many values, and therefore remain valid when one of the configurations is a boundary state whose value is set to zero.

\subsection{The same-edge identity}

Equations \eqref{eq:L-on-Ai} and \eqref{eq:L-on-Aj} give
\begin{align*}
 C_{ij}^{i}(A;-L_{ij}f)
 &=-4(11f_i-f_j-4f_\varnothing-4f_{ij})-(11f_j-f_i-4f_\varnothing-4f_{ij})\\
 &=-43f_i-10f_j+20f_\varnothing+20f_{ij}.
\end{align*}
Meanwhile,
\begin{align*}
 C_{ij}^{j}(A;f)-11C_{ij}^{i}(A;f)
 &=(f_i+4f_j-2f_\varnothing-2f_{ij})-11(4f_i+f_j-2f_\varnothing-2f_{ij})\\
 &=-43f_i-10f_j+20f_\varnothing+20f_{ij}.
\end{align*}
This proves \eqref{eq:id-same-edge}.

\subsection{An update on \texorpdfstring{$\{i,k\}$}{\{i,k\}}}

Assume first that $k\notin A$.  The operator $L_{ik}$ annihilates the values at $A$ and $Aj$, since their local $i,k$ state is $00$.  It acts on $Ai$ and $Aij$ by
\begin{align*}
 (L_{ik}f)(Ai)&=11f_i-f_k-4f_\varnothing-4f_{ik},\\
 (L_{ik}f)(Aij)&=11f_{ij}-f_{jk}-4f_j-4f_{ijk}.
\end{align*}
Consequently,
\begin{align}
 C_{ij}^{i}(A;-L_{ik}f)
 &=-44f_i+4f_k+16f_\varnothing+16f_{ik}-11f_{ij}+f_{jk}+4f_j+4f_{ijk}.
 \label{eq:appendix-ik-out-left}
\end{align}
On the other hand,
\begin{align*}
 C_{jk}^{k}(A;f)
 &=f_j+4f_k-2f_\varnothing-2f_{jk},\\
 4C_{jk}^{k}(Ai;f)
 &=4f_{ij}+16f_{ik}-8f_i-8f_{ijk},\\
 -9C_{ij}^{i}(A;f)
 &=-36f_i-9f_j+18f_\varnothing+18f_{ij}.
\end{align*}
Their sum is exactly \eqref{eq:appendix-ik-out-left}, proving \eqref{eq:id-ik-k-out}.

Now let $A=Bk$, with $B\cap\{i,j,k\}=\varnothing$.  The local $i,k$ states at $Bk,Bjk$ are $01$, and those at $Bik,Bijk$ are $11$.  Thus
\begin{align*}
 (L_{ik}f)(Bk)&=11f_k-f_i-4f_\varnothing-4f_{ik},\\
 (L_{ik}f)(Bjk)&=11f_{jk}-f_{ij}-4f_j-4f_{ijk},\\
 (L_{ik}f)(Bik)&=(L_{ik}f)(Bijk)=0.
\end{align*}
It follows that
\begin{align}
 C_{ij}^{i}(Bk;-L_{ik}f)
 &=- (L_{ik}f)(Bjk)+2(L_{ik}f)(Bk)\notag\\
 &=22f_k-2f_i-8f_\varnothing-8f_{ik}-11f_{jk}+f_{ij}+4f_j+4f_{ijk}.
 \label{eq:appendix-ik-in-left}
\end{align}
Expanding
\begin{equation*}
 4C_{jk}^{k}(B;f)+C_{jk}^{k}(Bi;f)-3C_{ij}^{i}(Bk;f)
\end{equation*}
gives the same expression \eqref{eq:appendix-ik-in-left}.  This proves \eqref{eq:id-ik-k-in}.

\subsection{An update on \texorpdfstring{$\{j,k\}$}{\{j,k\}}}

The remaining shared-endpoint identities may be checked in the same manner.  When $k\notin A$,
\begin{align*}
 (L_{jk}f)(Aj)&=11f_j-f_k-4f_\varnothing-4f_{jk},\\
 (L_{jk}f)(Aij)&=11f_{ij}-f_{ik}-4f_i-4f_{ijk},
\end{align*}
whereas the values at $A$ and $Ai$ are annihilated.  Therefore
\begin{align}
 C_{ij}^{i}(A;-L_{jk}f)
 &=-11f_j+f_k+4f_\varnothing+4f_{jk}+22f_{ij}-2f_{ik}-8f_i-8f_{ijk}.
 \label{eq:appendix-jk-out-left}
\end{align}
A direct expansion of
\begin{equation*}
 C_{ik}^{i}(A;f)+4C_{ik}^{i}(Aj;f)-3C_{ij}^{i}(A;f)
\end{equation*}
produces \eqref{eq:appendix-jk-out-left}, proving \eqref{eq:id-jk-k-out}.

If $A=Bk$, then
\begin{align*}
 (L_{jk}f)(Bk)&=11f_k-f_j-4f_\varnothing-4f_{jk},\\
 (L_{jk}f)(Bik)&=11f_{ik}-f_{ij}-4f_i-4f_{ijk},
\end{align*}
and the configurations $Bjk,Bijk$ have local state $11$.  Hence
\begin{align}
 C_{ij}^{i}(Bk;-L_{jk}f)
 &=-4(L_{jk}f)(Bik)+2(L_{jk}f)(Bk)\notag\\
 &=-44f_{ik}+4f_{ij}+16f_i+16f_{ijk}+22f_k-2f_j-8f_\varnothing-8f_{jk}.
 \label{eq:appendix-jk-in-left}
\end{align}
Expanding
\begin{equation*}
 4C_{ik}^{i}(B;f)+C_{ik}^{i}(Bj;f)-9C_{ij}^{i}(Bk;f)
\end{equation*}
gives \eqref{eq:appendix-jk-in-left}, proving \eqref{eq:id-jk-k-in}.

\subsection{Disjoint updates}

Let $\{k,l\}\cap\{i,j\}=\varnothing$.  The operator $L_{kl}$ acts only on the $k,l$ coordinates, while
\begin{equation*}
 f\longmapsto C_{ij}^{i}(A;f)
\end{equation*}
is a linear combination of four evaluations differing only in the $i,j$ coordinates.  The two operations commute.  More explicitly, if
$g(D)=C_{ij}^{i}(D;f)$, then applying $-L_{kl}$ to each of the four terms defining $g(A)$ gives
\begin{equation*}
 C_{ij}^{i}(A;-L_{kl}f)=(-L_{kl}g)(A),
\end{equation*}
which is \eqref{eq:id-disjoint}.  Interchanging $i,j$ proves all identities with superscript $j$.

\bigskip
\noindent{\bf Acknowledgements.}
Y. L. was supported by the Guangdong Basic and Applied Basic Research Foundation (2026A1515011707).
H.Z. gratefully acknowledges support from the NTU Research Scholarship.

\end{document}